\documentclass[11pt,letterpaper]{article}

\usepackage{amsmath,amsfonts,amsthm,amssymb,relsize,bm}  
\theoremstyle{plain}

\newcommand{\complex}{{\mathbb C}}
\newcommand{\real}{{\mathbb R}}

\numberwithin{equation}{section}

\newtheorem{thm}{Theorem}[section]
\newtheorem{lem}[thm]{Lemma}
\newtheorem{cor}[thm]{Corollary}

\theoremstyle{definition}  
\newtheorem{exam}{Example}  

\allowdisplaybreaks  

\newcounter{cond}

\newcommand{\tbullet}{\mathrel{\raise .4ex\hbox{\tiny$\bullet$}}} 

\newcommand{\rmtr}{\mathrm{tr\,}}

\newcommand{\escript}{\mathcal{E}}

\newcommand{\lscript}{\mathcal{L}}

\newcommand{\sscript}{\mathcal{S}}

\newcommand{\abar}{\overline{a}}
\newcommand{\bbar}{\overline{b}}

\newcommand{\ab}[1]{\left|#1\right|}
\newcommand{\doubleab}[1]{\left\|#1\right\|}
\newcommand{\brac}[1]{\left\{#1\right\}}
\newcommand{\paren}[1]{\left(#1\right)}
\newcommand{\sqbrac}[1]{\left[#1\right]}
\newcommand{\elbows}[1]{{\left\langle#1\right\rangle}}
\newcommand{\ket}[1]{{\left|#1\right>}}
\newcommand{\bra}[1]{{\left<#1\right|}}

\begin{document}

\title{MUTUALLY UNBIASED\\QUANTUM OBSERVABLES}
\author{Stan Gudder\\ Department of Mathematics\\
University of Denver\\ Denver, Colorado 80208\\
sgudder@du.edu}
\date{}
\maketitle

\begin{abstract}
We begin by defining mutually unbiased (MU) observables on a finite dimensional Hilbert space. We also consider the more general concept of parts of MU observables. The relationships between MU observables, value-complementary observables and two other conditions involving sequential products of observables are discussed. We next present a special motivating case of MU observables called finite position and momentum observables. These are atomic observables related by a finite Fourier transform. Finite position and momentum observables are employed to give examples of parts of MU observables that are value-complementary and those that are not value-complementary. Various open problems involving these concepts are presented. These problems mainly involve extending this work from sharp observables to unsharp observables.

\end{abstract}

\section{Basic Definitions}  
In this work we consider quantum systems represented by a finite-dimensional complex Hilbert space $H$. Although this is a strong restriction, the resulting framework is general enough to include theories of quantum computation, communication and information \cite{iva81,kw20,nc00}. We denote the set of linear operators on $H$ by $\lscript (H)$. An element $A\in\lscript (H)$ is \textit{positive} and we write $A\ge 0$ if $\elbows{\psi, A\psi}\ge 0$ for all $\psi\in H$. If $A-B\ge 0$ we write $B\le A$ or $A\ge B$ and if $0\le A\le I$ where $0,I$ are the zero and identity operators, respectively, we call $A$ an \textit{effect} \cite{bgl95,hz12,kra83,nc00}. An effect $A$ represents a two-outcome or yes-no experiment and when $A$ produces outcome yes we say that $A$ \textit{occurs}. We denote the set of effects on $H$ by $\escript (H)$. Then $(\escript (H),\le\,)$ becomes a partially ordered set with first and last elements $0,I$, respectively. For $A\in\escript (H)$ we call $A'=I-A$ the \textit{complement} of $A$ and, of course, $A+A'=I$. The effect $A'$ occurs if and only if the effect $A$ does not occur. We call an effect A \textit{sharp} if $A$ is a projection and we call $A$ \textit{atomic} if $A$ is a one-dimensional projection \cite{bgl95,hz12}.

An effect $\rho$ satisfying $\rmtr (\rho )=1$ is called a \textit{state} and states represent initial conditions of a quantum system. The set of states on $H$ is denoted by $\sscript (H)$ and for $A\in\escript (H)$, $\rho\in\sscript (H)$ the probability that $A$ occurs when the system is in state $\rho$ is defined as $\rmtr (\rho A)$. For $A,B\in\escript (H)$ we define their \textit{sequential product} by $A\circ B=A^{1/2}BA^{1/2}$, where $A^{1/2}$ is the unique positive square-root of $A$ \cite{gg02,gn01}. It is easy to check that $A\circ B\in\escript (H)$ which we interpret as the effect resulting from first measuring $A$ and then measuring $B$. It follows that $A$ can influence the measurement of $B$ but not vice versa. It can be shown that
$A\circ B=B\circ A$ if and only if $AB=BA$ \cite{gn01}.

An \textit{observable} for a quantum system is represented by a finite set of effects
$A=\brac{A_x\colon x\in\Omega _A}\subseteq\escript (H)$ where $\Omega _A$ is a finite set called the \textit{outcome space} for $A$ and
$\sum\limits _{x\in\Omega _A}A_x=I$. We interpret $A_x\in A$ as the effect that occurs when a measurement of $A$ results in the outcome $x$. If $\rho\in\sscript (H)$ we have that
\begin{equation*} 
\sum _{x\in\Omega _A}\rmtr (\rho A_x)=\rmtr (\rho )=1
\end{equation*}
so at least one of the effects $A_x$ occurs or equivalently, at least one of the outcomes $x\in\Omega _A$ results. The \textit{distribution} of an observable $A=\brac{A_x\colon x\in\Omega _A}$ in the state $\rho$ is the probability measure $\Phi _\rho ^A$ on $\Omega _a$ given by
$\Phi _\rho ^A(x)=\rmtr (\rho A_x)$. If $A=\brac{A_x\colon x\in\Omega _A}$, $B=\brac{B_y\colon y\in\Omega _B}$ are observables, we define their 
\textit{sequential product} $A\circ B$ to be the observable with outcomes space $\Omega _{A\circ B}=\Omega _A\times\Omega _B$ given by
\cite{gud120,gud220}
\begin{equation*} 
A\circ B=\brac{A_x\circ B_y\colon (x,y)\in\Omega _A\times\Omega _B}
\end{equation*}
so that $A\circ B_{(x,y)}=A_x\circ B_y$. The distribution of $A\circ B$ in the state $\rho$ becomes
\begin{equation*} 
\Phi _\rho ^{A\circ B}(x,y)=\rmtr (\rho A_x\circ B_y)=\rmtr\sqbrac{(A_x\circ\rho )B_y}
\end{equation*}
We also define the observable $B$ \textit{conditioned} by the observable $A$ to have outcome space $\Omega _B$ and to be given by
\cite{gud120,gud220}
\begin{equation*} 
(B\mid A)_y=\sum _{x\in\Omega _A}A_x\circ B_y=\sum _{x\in\Omega _A}A\circ B_{(x,y)}
\end{equation*}
We then have that
\begin{equation*} 
\Phi _\rho ^{(B\mid A)}(y)=\sum _{x\in\Omega _A}\Phi _\rho ^{A\circ B}(x,y)
\end{equation*}
It is easy to check that $A\circ B$ and $(B\mid A)$ are indeed observables. We say that an observable $A$ is \textit{sharp} if all of its effects $A_x$ are sharp and $A$ is \textit{atomic} if all of its effects $A_x$ are atomic. If $A$ and $B$ are sharp, then $A\circ B$ need not be sharp. Indeed, in this case, $A\circ B_{(x,y)}=A_xB_yA_x$ and simple examples show this may not be a projection. Also,
$(B\mid A)_y=\sum\limits _{x\in\Omega _A}A_xB_yA_x$ need not be a projection so $(B\mid A)$ need not be sharp.

An observable $B$ is \textit{part} of an observable $A$ and we write $B\subseteq A$ if there exists a surjection $f\colon\Omega _A\to\Omega _B$ such that $B_y=A_{f^{-1}(y)}$ for all $y\in\Omega _B$ \cite{fhl18,gud120}. In this case, $B_y=\sum\brac{A_x\colon f(x)=y}$. Notice that $B$ is indeed an observable because
\begin{equation*} 
\sum _{y\in\Omega _B}B_y=\sum _{y\in\Omega _B}\brac{A_x\colon f(x)=y}=\sum _{x\in\Omega _A}A_x=I
\end{equation*}
We also write $B=f(A)$. The distribution of $f(A)$ becomes
\begin{align*} 
\Phi _\rho ^{f(A)}(y)&=\rmtr\sqbrac{\rho f(A)_y}=\rmtr\sqbrac{\rho A_{f^{-1}(y)}}=\rmtr\sqbrac{\rho\sum\brac{A_x\colon f(x)=y}}\\
   &=\sum\brac{\rmtr (\rho A_x)\colon f(x)=y}
\end{align*}
Notice that if $B$ is part of a sharp observable, then $B$ is sharp. However, if $B\subseteq A$ and $A$ is atomic, then $B$ need not be atomic. Two observables $B,C$ \textit{coexist} if there exists an observable $A$ such that $B\subseteq A$, $C\subseteq A$ \cite{bgl95,fhl18,hz12,kra83}. In this case, $B$ and $C$ can be simultaneously measured by employing the observable $A$. It is not hard to show that $(B\mid A)$ and $A$ coexist \cite{gud120,gud220}.

Two orthonormal bases $\brac{\phi _i}$, $\brac{\psi _j}$ for $H$ are \textit{mutually unbiased} (MU) if $\ab{\elbows{\phi _i,\psi _j}}^2=1/d$ for all $i,j$ where $d=\dim H$ \cite{bgl95,iva81,wf84}. In this case, there is no bias between the two bases because the transition probabilities from one to the other are identical. Two atomic observables $A=\brac{\ket{\phi _i}\bra{\phi _i}\colon i=1,\ldots ,d}$,
$B=\brac{\ket{\psi _j}\bra{\psi _j}\colon j=1,\ldots ,d}$ are \textit{mutually unbiased} (MU) if $\brac{\phi _i}$, $\brac{\psi _j}$ are mutually unbiased. Two observables $A,B$ are \textit{part} MU if there exist MU observables $C,D$ such that $A\subseteq C$, $B\subseteq D$. Although part MU observables are sharp, they need not be atomic. Let $A=\brac{A_x\colon x\in\Omega _A}$,
$B=\brac{B_y\colon y\in\Omega _B}$ be observables on $H$ with $\ab{\Omega _A}=m$, $\ab{\Omega _B}=n$. We say that
$A$ and $B$ are \textit{value-complementary} \cite{hz12} if when $\Phi _\rho ^A(x)=1$, then $\Phi _\rho ^B(y)=1/n$ for every
$y\in\Omega _B$ and when $\Phi _\rho ^B(y)=1$, then $\Phi _\rho ^A(x)=1/m$ for every $x\in\Omega _A$. This condition says that if $A$ has an outcome $x$ with certainty in the state $\rho$ then $B$ is completely random (uncertain, undetermined) in the state $\rho$ and similarly, if $B$ has an outcome $y$ with certainty in the state $\rho$, then $A$ is completely random in the state $\rho$. This definition is motivated by similar properties possessed by complementary position and momentum observables in continuum quantum mechanics.

\begin{lem}    
\label{lem11}
Two observables $A=\brac{A_x\colon x\in\Omega _A}$, $B=\brac{B_y\colon y\in\Omega _B}$ are value-complementary if and only if
$\elbows{\psi ,A_x\psi}=1$ for $\doubleab{\psi}=1$ implies that $\elbows{\psi ,B_y\psi}=1/n$ for all $y\in\Omega _B$ and $\elbows{\psi ,B_y\psi}=1$ for 
$\doubleab{\psi}=1$ implies that $\elbows{\psi ,A_x\psi}=1/m$ for all $x\in\Omega _A$.
\end{lem}
\begin{proof}
If $A,B$ are value-complementary and $\elbows{\psi ,A_x\psi}=1$ for $\doubleab{\psi}=1$, letting $\rho =\ket{\psi}\bra{\psi}$ we have that
\begin{equation*}
\Phi _\rho ^A(x)=\rmtr (\rho A_x)=\elbows{\psi ,A_x\psi}=1
\end{equation*}
which implies that
\begin{equation*}
\elbows{\psi ,B_y\psi}=\rmtr(\rho B_y)=\Phi _\rho ^B(y)=\frac{1}{n}
\end{equation*}
for all $y\in\Omega _B$. Similarly, if $\elbows{\psi ,B_y\psi}=1$, then $\elbows{\psi ,A_x\psi}=1/m$ for all $x\in\Omega _A$. Conversely, Suppose the two conditions of the lemma hold and $\rmtr (\rho A_x)=1$ for $\rho\in\sscript (H)$. By the spectral theorem we can write $\rho =\sum\lambda _iP_i$,
$\lambda _i>0$ and $P_i=\ket{\phi _i}\bra{\phi _i}$. Then
\begin{equation*}
1=\rmtr (\rho A_x)=\rmtr\paren{\sum\lambda _i\ket{\phi _i}\bra{\phi _i}A_x}=\sum\lambda _i\elbows{\phi _iA_x\phi _i}
\end{equation*}
It follows that $\elbows{\phi _i,A_x\phi _i}=1$. Hence, $\elbows{\phi _i,B_y\phi _i}=1/n$ for all $i,y$. But then
\begin{equation*}
\rmtr (\rho B_y)=\sum\lambda _i\elbows{\phi _i,B_y\phi _i}=\sum\frac{1}{n}\,\lambda _i=\frac{1}{n}
\end{equation*}
The other condition is similar.
\end{proof}

\section{Value-Complementary Observables}  
We begin this section with two rather unusual conditions. Let  $A=\brac{A_x\colon x\in\Omega _A}$, $B=\brac{B_y\colon y\in\Omega _B}$ be two observables on $H$ where the cardinalities $\ab{\Omega _A}=m$, $\ab{\Omega _B}=n$. We define
\begin{list} {\hskip 1pc Condition (\arabic{cond})}{\usecounter{cond}
\setlength{\rightmargin}{\leftmargin}}
\item $A\circ B=\tfrac{1}{n}\,A$ and $B\circ A=\tfrac{1}{m}\,B$,
\item $(B\mid A)=\tfrac{1}{n}\,I$ and $(A\mid B)=\tfrac{1}{m}$\,I.
\end{list}\smallskip
Condition~(1) is shorthand for $A_x\circ B_y=\tfrac{1}{n}\,A_x$ and $B_y\circ A_x=\tfrac{1}{m}B_y$ for all $x\in\Omega _A$, $y\in\Omega _B$ and Condition~(2) is shorthand for $(B\mid A)_y=\tfrac{1}{n}\,I$ and $(A\mid B)_x=\tfrac{1}{m}\,I$ for all $x\in\Omega _A$, $y\in\Omega _B$. Condition~(1) says that $B$ uniformly attenuates $A$ and $A$ uniformly attenuates $B$. Condition~(2) says that $A$ and $B$ are conditionally random relative to each other. The next result gives a relationship between these two conditions and whether $A,B$ are value-complementary.

\begin{thm}    
\label{thm21}
{\rm{(i)}}\enspace Condition~(1) implies Condition~(2).
{\rm{(ii)}}\enspace If $A$ and $B$ are sharp, then Condition~(1) and Condition~(2) are equivalent.
{\rm{(iii)}}\enspace If $A$ and $B$ are sharp, then Condition~(1) or Condition~(2) imply that $A,B$ are value-complementary.
\end{thm}
\begin{proof}
(i)\enspace If Condition~(1) holds, we conclude that
\begin{align*}
(B\mid A)_y&=\sum\brac{A_x\circ B_y\colon x\in\Omega _A}=\sum\brac{\frac{1}{n}\,A_x\colon x\in\Omega _A}=\frac{1}{n}\,I
\intertext{and}
(A\mid B)_x&=\sum\brac{B_y\circ A_x\colon y\in\Omega _B}=\sum\brac{\frac{1}{m}\,B_y\colon y\in\Omega _B}=\frac{1}{m}\,I
\end{align*}
Hence, Condition~(2) holds.\newline
(ii)\enspace Suppose $A$ and $B$ are sharp. By (i) we have that Condition~(1) implies Condition~(2). If Condition~(2) holds, then for all
$y\in\Omega _B$ we obtain
\begin{equation}                
\label{eq21}
\sum\brac{A_xB_yA_x\colon x\in\Omega _A}=\sum\brac{A_x\circ B_y\colon x\in\Omega _A}=(B\mid A)_y=\frac{1}{n}\,I
\end{equation}
Multiplying \eqref{eq21} on the left by $A_z$ gives
\begin{equation*} 
A_z\circ B_y=A_zB_yA_z=\frac{1}{n}\,A_z
\end{equation*}
for all $z\in\Omega _A$, $y\in\Omega _B$. In a similar way, we obtain $B_y\circ A_x=\tfrac{1}{m}\,B_y$ for all $x\in\Omega _A$,
$y\in\Omega _B$ so Condition~(2) holds. We conclude that Conditions~(1) and (2) are equivalent.\newline
(iii)\enspace Suppose Condition~(1) holds. If $\elbows{\psi ,A_x\psi}=1$ for $\doubleab{\psi}=1$, we have that
\begin{equation*} 
\elbows{\psi ,A_x\circ B_y\psi}=\frac{1}{n}\elbows{\psi ,A_x\psi}=\frac{1}{n}
\end{equation*}
for all $y\in\Omega _B$. Since $A_x$ is sharp, $\elbows{\psi ,A_x\psi}=1$ for $\doubleab{\psi}=1$ implies that $A_x\psi =\psi$. Hence, for all
$y\in\Omega _B$ we obtain
\begin{equation*} 
\elbows{\psi ,B_y\psi}=\elbows{A_x\psi ,B_yA_x\psi}=\elbows{\psi ,A_x\circ B_y\psi}=\frac{1}{n}
\end{equation*}
In a similar way, if $\elbows{\psi ,B_y\psi}=1$ for $\doubleab{\psi}=1$, we obtain $\elbows{\psi ,A_x\psi}=1/m$ for all $x\in\Omega _A$. Applying
Lemma~\ref{lem11}, we conclude that $A,B$ are value-complementary. If Condition~(2) holds, then by (ii), Condition~(1) holds so again $A,B$ are value-complementary.
\end{proof}

We do not know whether the converse of Theorem~\ref{thm21}(i) or (ii) holds. Part (ii) of the next result strengthens Theorem~\ref{thm21}(iii).

\begin{thm}    
\label{thm22}
{\rm{(i)}}\enspace Condition~(1) holds if and only if for all $x\in\Omega _A$ and mutually orthogonal unit eigenvectors $\phi _1,\phi _2$ of $A_x$ with nonzero eigenvalues, we have that $\elbows{\phi _1,B_y\phi _2}=0$ and $\elbows{\phi _1,B_y\phi _1}=1/n$ for every $y\in\Omega _B$. Moreover, for mutually orthogonal unit eigenvectors $\psi _1,\psi _2$ of $B_y$ with nonzero eigenvalues, we have that $\elbows{\psi _1,A_x\psi _2}=0$ and
$\elbows{\psi _1,A_x\psi _1}=1/m$ for every $x\in\Omega _A$.
{\rm{(ii)}}\enspace If Condition~(1) holds, then $A,B$ are value-complementary.
\end{thm}
\begin{proof}
(i)\enspace Suppose Condition~(1) holds and $\phi$ is a unit eigenvector for $A_x$ with nonzero eigenvalue $\lambda$. It follows that
\begin{equation*}
\lambda\elbows{\phi ,B_y\phi}=\elbows{A_x^{1/2}\phi ,B_yA_x^{1/2}\phi}=\elbows{\phi ,A_x\circ B_y\phi}=\frac{1}{n}\elbows{\phi ,A_x\phi}=\frac{\lambda}{n}
\end{equation*}
Hence, $\elbows{\phi ,B_y\phi}=1/n$ for all $y\in\Omega _B$. Moreover, if $\phi _1,\phi _2$ are mutually orthogonal unit eigenvectors of $A_x$ with nonzero eigenvalues
$\lambda _1,\lambda _2$, respectively, we obtain
\begin{align*}
\lambda _1^{1/2}\lambda _2^{1/2}\elbows{\phi _1,B_y\phi _2}&=\elbows{A_x^{1/2}\phi _1,B_yA_x^{1/2}\phi _2}=\elbows{\phi _1,A_x\circ B_y\phi _2}\\
   &=\frac{1}{n}\elbows{\phi _1,A_x\phi _2}=\frac{\lambda _2}{n}\elbows{\phi _1,\phi _2}=0
\end{align*}
Hence, $\elbows{\phi _1,B_y\phi _2}=0$ for all $y\in\Omega _B$. The corresponding result for $B_y$ is similar. Conversely, suppose the theorem's statements hold. Letting
$\brac{\phi _i}$ be an orthonormal basis of eigenvectors for $A_x$, if $\alpha\in H$ is an arbitrary vector, we can write $\alpha =\sum\alpha _i\phi _i$. Letting
$A_x\phi _i=\lambda _i\phi _i$, we obtain
\begin{align*}
\elbows{\alpha ,\frac{1}{n}\,A_x\alpha}&=\elbows{\sum\alpha _j\phi _j,\frac{1}{n}\sum\alpha _i\lambda _i\phi _i}=\frac{1}{n}\sum\ab{\alpha _i}^2\lambda _i\\
   &=\sum\ab{\alpha _i}^2\lambda _i\elbows{\phi _i,B_y\phi _i}=\elbows{\sum\alpha _j\lambda _j^{1/2}\phi _j,\sum\alpha _i\lambda _i^{1/2}B_y\phi _i}\\
   &=\elbows{\sum\alpha _jA_x^{1/2}\phi _j,\sum\alpha _iB_yA_x^{1/2}\phi _i}\\
   &=\elbows{\sum\alpha _j\phi _j,\sum\alpha _iA_x\circ B_y\phi _i}=\elbows{\alpha ,A_x\circ B_y\alpha}
\end{align*}
Hence, $A_x\circ B_y=\tfrac{1}{n}\,A_x$ and similarly, $B_y\circ A_x=\tfrac{1}{m}\,B_y$ for all $x\in\Omega _A$, $y\in\Omega _B$.\newline
(ii)\enspace Suppose Condition~(1) holds and $\elbows{\phi _,A_x\phi}=1$ for $x\in\Omega _A$ and unit vector $\phi$. Then by Schwarz's inequality we obtain
\begin{equation*}
1=\elbows{\phi ,A_x\phi}\le\doubleab{\phi}\,\doubleab{A_x\phi}=\doubleab{A_x\phi}\le 1
\end{equation*}
Hence, $\doubleab{A_x\phi}=1$ and we have equality $\elbows{\phi ,A_x\phi}=\doubleab{\phi}\,\doubleab{A_x\phi}$. It follows that $A_x\phi =c\phi$, $c\in\complex$. Since
$\ab{c}=\doubleab{c\phi}=1$ and $\phi$ is an eigenvector of $A_x$ we conclude that $c=1$. By Part~(i) we have that $\elbows{\phi ,B_y\phi}=1/n$ for every $y\in\Omega _B$. Similarly,  if $\elbows{\psi ,B_y\psi}=1$ for $y\in\Omega _B$ and unit vector $\psi$, then $\elbows{\psi ,A_x\psi}=1/m$ for every $x\in\Omega _A$. Applying Lemma~\ref{lem11}, we conclude that $A,B$ are value-complementary.
\end{proof}

When $A$ and $B$ are atomic, we obtain the following stronger results than Theorems~\ref{thm21} and \ref{thm22}.

\begin{thm}    
\label{thm23}
If $A=\brac{A_x\colon x\in\Omega _A}$, $B=\brac{B_y\colon y\in\Omega _B}$ are atomic observables, then the following statements are equivalent.
{\rm{(i)}}\enspace $A,B$ are MU.
{\rm{(ii)}}\enspace $A,B$ are value-complementary.
{\rm{(iii)}}\enspace Condition~(1) holds.\newline
{\rm{(iv)}}\enspace Condition~(2) holds.
\end{thm}
\begin{proof}
Let $A_x=\ket{\phi _x}\bra{\phi _x}$, $B_y=\ket{\psi _y}\bra{\psi _y}$ where $\brac{\phi _x}$, $\brac{\psi _y}$ are orthonormal bases for $H$. If $\brac{\phi _x}$, $\brac{\psi _y}$ are MU, then $\ab{\elbows{\phi _x,\psi _y}}^2=1/d$ for all $x\in\Omega _A$, $y\in\Omega _B$ where $d=\dim H$. If $\elbows{\psi ,A_x\psi}=1$ for $\doubleab{\psi}=1$, then
$\ab{\elbows{\psi ,\phi _x}}^2=1$. Therefore, $\psi =c\phi _x$ where $\alpha\in\complex$ with $\ab{c}=1$. It follows that $\ab{\elbows{\psi ,\psi _y}}=1/d$ so that 
$\elbows{\psi ,B_y\psi}=1/d$ for all $y\in\Omega _B$. Similarly, if $\elbows{\psi ,B_y\psi}=1$ then $\elbows{\psi ,A_x\psi}=1/d$ for all $x\in\Omega _A$. By Lemma~\ref{lem11}, $A$ and $B$ value-complementary. Conversely, suppose $A,B$ are value-complementary. Since
\begin{equation*}
\elbows{\phi _x,A_x\phi _x}=\doubleab{\phi _x}^2=1
\end{equation*}
we have that
\begin{equation*}
\ab{\elbows{\phi _x,\psi _y}}^2=\elbows{\phi _x,B_y\phi _x}=\frac{1}{d}
\end{equation*}
for all $y\in\Omega _B$. Hence, $\brac{\phi _x}$, $\brac{\psi _y}$ are MU. We conclude that (i) and (ii) are equivalent. Now suppose that (i) holds. We then have 
\begin{align*}
A_x\circ B_y&=A_xB_yA_x=\ket{\phi _x}\bra{\phi _x}\,\ket{\psi _y}\bra{\psi _y}\,\ket{\phi _x}\bra{\phi _x}
   =\ab{\elbows{\phi _x,\psi _y}}^2\ket{\phi _i}\bra{\phi _i}\\
   &=\frac{1}{d}\,A_x
\end{align*}
for all $x\in\Omega _A$, $y\in\Omega _B$. In a similar way, $B_y\circ A_x=\tfrac{1}{d}\,B_y$ for all $x\in\Omega _A$, $y\in\Omega _B$. Hence, (i) implies (iii). It follows from Theorem~\ref{thm21} that (iii) and (iv) are equivalent and that (iii) or (iv) imply (ii). Hence, the four statements are equivalent.
\end{proof}

\section{Finite Position and Momentum Observables}  
For this section we work in the standard Hilbert space $H=\complex ^N$. For $\psi =(\alpha _0,\alpha _1,\ldots ,\alpha _{N-1})\in H$, the \textit{finite Fourier transform} is the unitary operator $F$ on $H$ defined by \cite{hz12}
\begin{equation*} 
(F\psi )_j=\frac{1}{\sqrt{N\,}}\sum _{n=0}^{N-1}e^{-2\pi ijn/N}\alpha _n,i=\sqrt{-1\,},j=0,1,\ldots ,N-1
\end{equation*}
It follows that
\begin{equation*}
(F^*\psi )_j=\frac{1}{\sqrt{N\,}}\sum _{n=0}^{N-1}e^{2\pi ijn/N}\alpha _n, j=0,1,\ldots ,N-1
\end{equation*}
Let $\brac{\phi _k\colon k=0,1,\ldots ,N-1}$ be the standard basis $\phi _k=(0,\ldots ,0,1,0,\ldots ,0)$ where 1 is in the kth position, $k=0,1,\ldots ,N-1$. The jth component of $F\phi _k$ becomes
\begin{equation*}
(F\phi _k)_j=\frac{1}{\sqrt{N\,}}\sum _{n=0}^{N-1}e^{-2\pi ijn/N}(\phi _k)_n=\frac{1}{\sqrt{N\,}}\sum _{n=0}^{N-1}e^{-2\pi ijn/N}\delta _{jk}
   =\frac{1}{\sqrt{N\,}}\,e^{-2\pi ijk/N}
\end{equation*}
Relative to the standard basis, $F$ has the matrix representation
\begin{equation*}
\begin{bmatrix}F_{mn}\end{bmatrix}=\frac{1}{\sqrt{N\,}}\sqbrac{e^{-2\pi imn/\sqrt{N\,}}\,}
\end{equation*}

\begin{exam}  
Letting $N=2$, we have
\begin{align*}
F\phi _0&=\frac{1}{\sqrt{2\,}}\begin{bmatrix}1\\1\end{bmatrix},\ F\phi _1=\frac{1}{\sqrt{2\,}}\begin{bmatrix}1\\e^{-2\pi i/2}\end{bmatrix}
  =\frac{1}{\sqrt{2\,}}\begin{bmatrix}1\\-1\end{bmatrix},\\\noalign{\smallskip}
  F&=\frac{1}{\sqrt{2\,}}\begin{bmatrix}1&1\\1&-1\end{bmatrix}\hskip 20pc\square
\end{align*}
\end{exam}

\begin{exam}  
Letting $N=4$, we have
\begin{align*}
F\phi _0&=\frac{1}{2}\begin{bmatrix}1\\1\\1\\1\end{bmatrix},\ F\phi _1
  =\frac{1}{2}\begin{bmatrix}1\\e^{-2\pi i/4}\\e^{-2\pi i2/4}\\e^{-2\pi i3/4}\end{bmatrix}
  =\frac{1}{2}\begin{bmatrix}1\\-i\\-1\\i\end{bmatrix}\\\noalign{\smallskip}
  F\phi _2&=\frac{1}{2}\begin{bmatrix}1\\e^{-2\pi i2/4}\\e^{-2\pi i2\tbullet 2/4}\\e^{-2\pi i2\tbullet 3/4}\end{bmatrix}
  =\frac{1}{2}\begin{bmatrix}1\\-1\\1\\-1\end{bmatrix},
  F\phi _3=\frac{1}{2}\begin{bmatrix}1\\e^{-2\pi i3/4}\\e^{-2\phi i3\tbullet 2/4}\\e^{-2\pi i3\tbullet 3/4}\end{bmatrix}
  =\frac{1}{2}\begin{bmatrix}1\\i\\-1\\-i\end{bmatrix}\\\noalign{\smallskip}
  &F=\tfrac{1}{2}\begin{bmatrix}1&1&1&1\\1&-i&-1&i\\1&-1&1&-1\\1&i&-1&-i\end{bmatrix}\hskip 16pc\square
\end{align*}
\end{exam}

Define the atomic effects $Q_j\in\escript (H)$, $j=0,1,\ldots ,N-1$, given by $Q_j=\ket{\phi _j}\bra{\phi _j}$. The
\textit{finite position observable} is the atomic observable $Q=\brac{Q_0,Q_1,\ldots ,Q_{N-1}}$. Let $P_j$, $j=0,1,\ldots ,N-1$ be the atomic effects $P_j=FQ_jF^*$. The \textit{finite momentum observable} is the atomic observable $P=\brac{P_0,P_1,\ldots ,P_{N-1}}$. Notice that
\begin{equation*}
P_j=F\ket{\phi _j}\bra{\phi _j}F^*=\ket{F\phi _j}\bra{F\phi _j},\quad j=0,1,\ldots ,N-1
\end{equation*}

\begin{exam}  
For $N=2$ we can represent $Q_j$ by the matrices
\begin{equation*}
Q_0=\begin{bmatrix}1&0\\0&0\end{bmatrix},\ Q_1=\begin{bmatrix}0&1\\0&0\end{bmatrix}
\end{equation*}
and $P_j$ by the matrices
\begin{align*}
P_0&=FQ_0F^*=\frac{1}{2}\begin{bmatrix}1&1\\1&-1\end{bmatrix}\ \begin{bmatrix}1&0\\0&0\end{bmatrix}\ 
   \begin{bmatrix}1&1\\1&-1\end{bmatrix}=\frac{1}{2}\begin{bmatrix}1&1\\1&1\end{bmatrix}\\
   P_1&=FQ_1F^*=\frac{1}{2}\begin{bmatrix}1&1\\1&-1\end{bmatrix}\ \begin{bmatrix}0&0\\0&1\end{bmatrix}\ 
   \begin{bmatrix}1&1\\1&-1\end{bmatrix}=\frac{1}{2}\begin{bmatrix}1&-1\\-1&1\end{bmatrix}\hskip 4pc\square
\end{align*}
\end{exam}

\begin{exam}  
For $N=4$ we can represent $Q_j$ by the matrices
\begin{align*}
Q_0&=\begin{bmatrix}1&0&0&0\\0&0&0&0\\0&0&0&0\\0&0&0&0\end{bmatrix},\ 
Q_1=\begin{bmatrix}0&0&0&0\\0&1&0&0\\0&0&0&0\\0&0&0&0\end{bmatrix},\ 
Q_2=\begin{bmatrix}0&0&0&0\\0&0&0&0\\0&0&1&0\\0&0&0&0\end{bmatrix},\\\noalign{\smallskip}
Q_3&=\begin{bmatrix}0&0&0&0\\0&0&0&0\\0&0&0&0\\0&0&0&1\end{bmatrix}
\end{align*}
and $P_j$ by the matrices
\begin{align*} 
P_0&=\frac{1}{4}\begin{bmatrix}1\\1\\1\\1\end{bmatrix}{\raisebox{1.75pc}{$\begin{bmatrix}1&1&1&1\end{bmatrix}$}}
  =\frac{1}{4}\begin{bmatrix}1&1&1&1\\1&1&1&1\\1&1&1&1\\1&1&1&1\end{bmatrix},\\ 
  P_1&=\frac{1}{4}\begin{bmatrix}1\\-i\\-1\\i\end{bmatrix}{\raisebox{1.75pc}{$\begin{bmatrix}1&i&-1&-i\end{bmatrix}$}}
    =\frac{1}{4}\begin{bmatrix}1&i&-1&-i\\-i&1&i&-1\\-1&-i&1&i\\i&-1&-i&1\end{bmatrix},\\
P_2&=\frac{1}{4}\begin{bmatrix}1\\-1\\1\\-1\end{bmatrix}{\raisebox{1.75pc}{$\begin{bmatrix}1&-1&1&-1\end{bmatrix}$}}
  =\frac{1}{4}\begin{bmatrix}1&-1&1&-1\\-1&1&-1&1\\1&-1&1&-1\\-1&1&-1&1\end{bmatrix},\\ 
  P_3&=\frac{1}{4}\begin{bmatrix}1\\i\\-1\\-i\end{bmatrix}{\raisebox{1.75pc}{$\begin{bmatrix}1&-i&-1&i\end{bmatrix}$}}
    =\frac{1}{4}\begin{bmatrix}1&-i&-1&i\\i&1&-i&-1\\-1&i&1&-i\\-i&-1&i&1\end{bmatrix}\hskip 3pc\square
\end{align*}
\end{exam}

\begin{lem}    
\label{lem31}
{\rm{(i)}}\enspace If $a,b\in\escript (H)$ and $a$ is atomic, then $a\circ b=\rmtr (ab)a$.
{\rm{(ii)}}\enspace For all $j,k=0,1,\ldots ,N-1$ we have that $Q_j\circ P_k=\tfrac{1}{N}\,Q_j$ and $P_j\circ Q_k=\tfrac{1}{N}\,P_j$.
\end{lem}
\begin{proof}
(i)\enspace Since $a\circ b\le a$ and $a$ is atomic, we have that $a\circ b=\lambda a$ for some $\lambda\in\sqbrac{0,1}$. Taking the trace gives
\begin{equation*} 
\lambda =\rmtr (\lambda a)=\rmtr (a\circ b)=\rmtr (aba)=\rmtr (ab)
\end{equation*}
(ii)\enspace
Applying (i) we obtain $Q_j\circ P_k=\rmtr (Q_jP_k)Q_j$. Since
\begin{align*}
\rmtr (Q_jP_k)&=\rmtr (Q_jFQ_kF^*)=\elbows{\phi _j,FQ_kF^*\phi _j}=\elbows{F^*Q_j,Q_kF^*\phi _j}\\
   &=\elbows{Q_kF^*\phi _j,Q_kF^*\phi _j}=\ab{F_{kk}^*}^2=\frac{1}{N}
\end{align*}
we have that $Q_j\circ P_k=\tfrac{1}{N}\,Q_j$ for all $j,k=0,1,\ldots ,N-1$. Similarly $P_j\circ Q_k=\tfrac{1}{N}\,P_j$ for all $j,k=0,1,\ldots ,N-1$.
\end{proof}

Applying Theorem~\ref{thm22} we obtain the following

\begin{cor}    
\label{cor32}
The observables $Q,P$ are MU and value-complementary. Moreover $(Q\mid P)=(P\mid Q)=\tfrac{1}{N}\,I$.
\end{cor}

Since $Q,P$ are MU, any parts $Q'\subseteq Q$, $P'\subseteq P$ are part MU. This is illustrated in the following examples.

\begin{exam}  
For $N=4,$ let $Q'_0=Q_0+Q_1$, $Q'_1=Q_2+Q_3$, $P'_0=P_0+P2$, $P'_1=P_1+P_3$. Then $Q'=\brac{Q'_0,Q'_1}\subseteq Q$,
$P'=\brac{P'_0,P'_1}\subseteq P$ are part MU. We have that
\begin{align*}
Q'_0&=\begin{bmatrix}1&0&0&0\\0&1&0&0\\0&0&0&0\\0&0&0&0\end{bmatrix},\qquad
Q'_1=\begin{bmatrix}0&0&0&0\\0&0&0&0\\0&0&1&0\\0&0&0&1\end{bmatrix}\\
P'_0&=\frac{1}{2}\begin{bmatrix}1&0&1&0\\0&1&0&1\\1&0&1&0\\0&1&0&1\end{bmatrix},\qquad
P'_1=\frac{1}{2}\begin{bmatrix}1&0&-1&0\\0&1&0&-1\\-1&0&1&0\\0&-1&0&1\end{bmatrix}
\end{align*}
The sequential products become
\begin{align*} 
Q'_0\circ P'_0&=Q'_0P'_0Q'_0=\frac{1}{2}\begin{bmatrix}1&0&0&0\\0&1&0&0\\0&0&0&0\\0&0&0&0\end{bmatrix}=\frac{1}{2}\,Q'_0\\
\intertext{and}
Q'_0\circ P'_1&=Q'_0(I-P'_0)Q'_0=Q'_0-Q'_0P'_0Q'_0=Q'_0-\frac{1}{2}\,Q'_0=\frac{1}{2}\,Q'_0
\end{align*}
In a similar way, we obtain $Q'_1\circ P'_0=\tfrac{1}{2}\,Q'_1$, $Q'_1\circ P'_1=\tfrac{1}{2}\,Q'_1$. We also have
\begin{equation*}
P'_0\circ Q'_0=P'_0Q'_0P'_0=\frac{1}{4}\begin{bmatrix}1&0&1&0\\0&1&0&1\\1&0&1&0\\0&1&0&1\end{bmatrix}=\frac{1}{2}\,P'_0
\end{equation*}
and as before $P'_0\circ Q'_1=\tfrac{1}{2}\,P'_0$. Moreover, we obtain
\begin{equation*} 
P'_1\circ Q'_0=P'_1Q'_0P'_1=\frac{1}{2}\begin{bmatrix}1&0&-1&0\\0&1&0&-1\\-1&0&1&0\\0&-1&0&1\end{bmatrix}=\frac{1}{2}\,P'_1
\end{equation*}
and as before $P'_1\circ Q'_1=\tfrac{1}{2}\,P'_1$. We conclude that $Q'_j\circ P'_k=\tfrac{1}{2}\,Q'_j$ and
$P'_k\circ Q'_j=\tfrac{1}{2}\,P'_k$ for $j,k=0,1$. It follows from Theorem~\ref{thm21} that $Q',P'$ are value-complementary.\hfill\qedsymbol
\end{exam}
 
\begin{exam}  
For $N=4,$ let $Q'_0=Q_0+Q_1$, $Q'_1=Q_2+Q_3$, as in Example~5, but we now let $P''_0=P_0+P_1$, $P''_1=P_2+P_3$. We again have that $Q'=\brac{Q'_0,Q'_1}\subseteq Q$, $P''=\brac{P''_0,P''_1}\subseteq P$ are part MU. In this case, $P''_0$ and $P''_1$ are given by
\begin{align*}
P''_0&=\frac{1}{4}\begin{bmatrix}2&1+i&0&1-i\\1-i&2&1+i&0\\0&1-i&2&1+i\\1+i&0&1-i&2\end{bmatrix},\\\noalign{\smallskip}
P''_1&=\frac{1}{4}\begin{bmatrix}2&-1-i&0&-1+i\\-1+i&2&-1-i&0\\0&-1+i&2&-1-i\\-1-i&0&-1+i&2\end{bmatrix}
\end{align*}
The sequential product becomes
\begin{equation*} 
Q'_0\circ P''_0=Q'_0P''_0Q'_0=\frac{1}{4}\begin{bmatrix}2&1+i&0&0\\1-i&2&0&0\\0&0&0&0\\0&0&0&0\end{bmatrix}
\end{equation*}
so $Q'_0\circ P''_0\ne\tfrac{1}{2}\,Q'_0$ which is different from Example~5. In a similar way $Q'_0\circ P''_1\ne\tfrac{1}{2}\,Q'_0$,
$P''_0Q'_0\ne\tfrac{1}{2}\,P''_0$ and $P''_1\circ Q'_1\ne\tfrac{1}{2}\,P''_1$. Hence, $Q',P''$ do not satisfy Condition~(1). We now show that $Q',P''$ do not satisfy Condition~(2). (Of course, it would follow that $Q',P''$ do not satisfy Condition~(1).) Indeed, since
\begin{equation*} 
Q'_0\circ P''_0=\frac{1}{4}\begin{bmatrix}0&0&0&0\\0&0&0&0\\0&0&2&1+i\\0&0&1-i&2\end{bmatrix}
\end{equation*}
we obtain
\begin{equation*} 
Q'_0\circ P''_0+Q'_0P''_0=\frac{1}{4}\begin{bmatrix}2&1+i&0&0\\1-i&2&0&0\\0&0&2&1+i\\0&0&1-i&2\end{bmatrix}\ne\frac{1}{2}\,I
\end{equation*}
Hence, $Q',P''$ do not satisfy Condition~(2). We now show that $Q',P''$ are not value-complementary. Although MU observables are value-complementary, this would show that part MU observables need not be value-complementary. Suppose $\elbows{\psi ,Q'_0\psi}=1$ with $\doubleab{\psi}=1$. It follows that $\psi =(a,b,0,0)$ where $a,b\in\complex$ with $\ab{a}^2+\ab{b}^2=1$. We then have
\begin{equation*} 
P''_0\psi=\frac{1}{4}\begin{bmatrix}2a+(1+i)b\\(1-i)a+2b\\(1-i)b\\(1+i)a\end{bmatrix}
\end{equation*}
Hence,
\begin{align*}
\elbows{\psi ,P''_0\psi}&=\frac{1}{4}\sqbrac{2\ab{a}^2+(1+i)\abar\,b+(1-i)a\bbar+2\ab{b}^2}\\
   &=\frac{1}{4}\sqbrac{2+(1+i)\abar\,b+(1-i)a\bbar\,}
\end{align*}
In general, $\elbows{\psi ,P''_0\psi}\ne 1/2$. For example, if $a=b=1/\sqrt{2\,}$ we obtain
\begin{equation*} 
\elbows{\psi ,P''_0\psi}=\frac{1}{4}\sqbrac{2+\frac{1}{2}(1+i)+\frac{1}{2}(1-i)}=\frac{3}{4}
\end{equation*}
We conclude that $Q',P''$ are not value-complementary.\hfill\qedsymbol
\end{exam}  

\section{Concluding Remarks}  
Most of our work here has involved sharp observables. Can any of these results be extended to unsharp observables? In particular, can we extend the definition of MU observables to nonatomic or unsharp observables? Is there a relationship between these more general MU observables and value-complementary observables? What about parts of these more general MU observables? In this section, we begin a study of these questions.

Let $A=\brac{A_x\colon x\in\Omega _A}$, $B=\brac{B_y\colon y\in\Omega _B}$ be observables on $H$ and let $\dim H=d$,
$\ab{\Omega _A}=m$, $\ab{\Omega _B}=n$. We say that $A,B$ are \textit{generalized} MU \textit{observables} if $\rmtr (A_xB_y)=\alpha$ for all $x\in\Omega _A$, $y\in\Omega _B$ where $\alpha\in\real$. Notice that when this condition holds, then $\alpha =d/mn$. This is because
\begin{equation*} 
d=\sum _{x,y}\rmtr (A_xB_y)=\sum _{x,y}\alpha =mn\alpha
\end{equation*}
The next result shows that this definition reduces to the usual definition of MU observables when $A$ and $B$ are atomic.

\begin{thm}    
\label{thm41}
{\rm{(i)}}\enspace If $A_x=\ket{\phi _x}\bra{\phi _x}$ and $B_y=\ket{\psi _y}\bra{\psi _y}$ are atomic then $A,B$ are generalized MU if and only if $A,B$ are MU.
{\rm{(ii)}}\enspace If $A,B$ satisfy Condition~(1) then $A,B$ are generalized MU.
\end{thm}
\begin{proof}
(i)\enspace If $A,B$ are MU we have
\begin{equation*} 
\rmtr (A_xB_y)=\ab{\elbows{\phi _x,\psi _y}}^2=\frac{1}{d}
\end{equation*}
for all $x\in\Omega _A$, $y\in\Omega _B$. Conversely, if $A,B$ are generalized MU then
\begin{equation*}
\ab{\elbows{\phi _x,\psi _x}}^2=\rmtr (A_xB_y)=\frac{d}{mn}=\frac{1}{d}
\end{equation*}
(ii)\enspace If $A,B$ satisfy Condition~(1), then $A_x\circ B_y=\tfrac{1}{n}\,A_x$ and $B_y\circ A_x=\tfrac{1}{m}\,B_y$ for all
$x\in\Omega _A$, $y\in\Omega _B$. We then obtain
\begin{equation*}
\frac{1}{n}\,\rmtr (A_x)=\rmtr (A_x\circ B_y)=\rmtr (A_xB_y)=\rmtr (B_yA_x)=\rmtr (B_y\circ A_x)=\frac{1}{m}\,\rmtr (B_y)
\end{equation*}
for all $x\in\Omega _A$, $y\in\Omega _B$. Since $\rmtr (A_xB_y)=\tfrac{1}{n}\,\rmtr (A_x)$ we obtain
\begin{equation*}
\rmtr (B_y)=\sum _x\rmtr (A_xB_y)=\frac{1}{n}\sum _x\rmtr (A_x)=\frac{1}{n}\,\rmtr (I)=\frac{d}{n}
\end{equation*}
Hence,
\begin{equation*}
\rmtr (A_xB_y)=\frac{1}{m}\,\rmtr (B_y)=\frac{d}{mn}
\end{equation*}
for all $x\in\Omega _A$, $y\in\Omega _B$. Therefore, $A,B$ are generalized MU.
\end{proof}

The next result characterizes when part MU observables are generalized MU observables.

\begin{lem}    
\label{lem42}
If $f(A)$, $g(B)$ are parts of the MU observables $A,B$, then $f(A)$, $g(B)$ are generalized MU observables if and only if
\begin{equation*}
\ab{f^{-1}(r)}\ab{g^{-1}(s)}=\ab{f^{-1}(r')}\ab{g^{-1}(s')}
\end{equation*}
for all $r,r'\in\Omega _{f(A)}$ and $s,s'\in\Omega _{g(B)}$.
\end{lem}
\begin{proof}
This follows from the equation:
\begin{align*}
\rmtr\sqbrac{f(A)_r,g(B)_s}&=\rmtr\sqbrac{\sum\brac{A_x\colon f(x)=r}\sum\brac{B_y\colon g(y)=s}}\\
  &=\sum _{f(x)=r}\sum _{g(y)=s}\rmtr (A_xB_y)=\sum _{f(x)=r}\sum _{g(y)=s}\ab{\elbows{\phi _x,\psi _y}}^2\\\noalign{\smallskip}
  &=\sum _{f(x)=r}\sum _{g(y)=s}\frac{1}{d}=\frac{\ab{f^{-1}(r)}\ab{g^{-1}(s)}}{d}
\end{align*}
\end{proof}

Applying Lemma~\ref{lem42} we see that the observables $Q',P'$ of Example~5 and the observables $Q',P''$ of Example~6 are generalized MU. Also, Example~6 shows that generalized MU observables need not satisfy Condition~(1) or Condition~(2) or be value-complementary.

\begin{exam}  
The generalized MU observables of Examples~5 and 6 are sharp. We now give a simple example of a pair of generalized MU observables that are unsharp. Let $A$ be the unsharp observable given by $A=\brac{\tfrac{1}{d}\,I,\tfrac{1}{d}\,I,\ldots ,\tfrac{1}{d}\,I}$ where there are
$d$ effects and let $B$ be an unsharp observable given by $B=\brac{B_1,B_2,\ldots ,B_m}$ where $\rmtr (B_j)=a$ for $j=1,2,\ldots ,m$. Since $\rmtr (A_iB_j)=a/d$ for all $i=1,2,\ldots ,d$ and $j=1,2,\ldots ,m$, we have that $A,B$ are generalized MU
observables.\hfill\qedsymbol
\end{exam}  

Another possible approach for defining a pair of generalized MU observables is to use Condition~(1). According to
Theorem~\ref{thm41}(ii) this has the advantage of being a stronger condition than our present definition. In fact, as noted in the previous paragraph, Condition~(1) is a strictly stronger definition. Moreover, by Theorem~\ref{thm23}, Condition~(1) also reduces to the usual criterion for atomic observables. Finding the best definition for generalized MU observables requires further research. A simple example of unsharp observables satisfying Condition~(1) are the \textit{trivial} observables $A,B$ where $A_x=\tfrac{1}{m}\,I$, $B_y=\tfrac{1}{n}\,I$ for all $x\in\Omega _A$, $y\in\Omega _B$. Notice that if $A,B$ satisfy Condition~(1), then they are trivial if and only if $A_x$ or $B_y$ is invertible for some $x\in\Omega _A$, $y\in\Omega _B$. Indeed, if $A$ and $B$ are trivial, then clearly they are invertible. Conversely, suppose $A_x$ is invertible for some $x\in\Omega _A$. Then since $A_x^{1/2}B_yA_x^{1/2}=\tfrac{1}{n}\,A_x$, multiplying on both sides with $A_x^{-1/2}$ gives $B_y=\tfrac{1}{n}\,I$ for all $y\in\Omega _B$. Moreover, since $B_y^{1/2}A_xB_y^{1/2}=\tfrac{1}{m}\,B_y$ we obtain
\begin{equation*}
\frac{1}{n}\,A_x=B_y^{1/2}A_xB_y^{1/2}=\frac{1}{mn}\,I
\end{equation*}
Hence, $A_x=\tfrac{1}{m}\,I$ for $x\in\Omega _A$. If $B_y$ is invertible for some $y\in\Omega _B$ the result is similar.

\end{document}